\documentclass[11pt,a4paper]{article}

\usepackage[lmargin=1.0in,rmargin=1.0in,bottom=1.0in,top=1.0in,twoside=False]{geometry}
\usepackage{microtype}
\usepackage{libertine}

\usepackage{enumerate}
\usepackage{enumitem}
\usepackage{hyperref}
\usepackage{amsmath,amssymb,amsthm}
\usepackage{tikz}
\usepackage{algorithm}
\usepackage{algpseudocode}
\usepackage[capitalise]{cleveref}
\usepackage{fullpage}
\usepackage{todonotes}
\usepackage{authblk}
\usepackage{xspace}
\usepackage[normalem]{ulem}

\usepackage[absolute]{textpos}
\usepackage[all=normal,bibliography=tight]{savetrees}

\newtheorem{theorem}{Theorem}[section]
\newtheorem{lemma}[theorem]{Lemma}
\newtheorem{corollary}[theorem]{Corollary}

\theoremstyle{definition}

\theoremstyle{remark}

\newcommand{\yes}{\textsc{Yes}}
\newcommand{\no}{\textsc{No}}
\newcommand{\Oh}{\mathcal{O}}
\newcommand{\terms}{\mathcal{T}}
\newcommand{\weight}{\omega}
\newcommand{\inst}{\mathcal{I}}
\newcommand{\gdpc}{\textsc{GDPC}}
\newcommand{\pairs}{\mathcal{C}}
\newcommand{\bundles}{\mathcal{B}}
\newcommand{\maxarity}{b}

\title{On weighted graph separation problems and flow-augmentation%
\thanks{This research is a part of a project that have received funding from the European Research Council (ERC)
under the European Union's Horizon 2020 research and innovation programme
Grant Agreement 714704 (M. Pilipczuk). Eun Jung Kim is supported by the grant from French National Research Agency under JCJC program (ASSK: ANR-18-CE40-0025-01).}
}

\date{}

\author[1]{Eun Jung Kim}
\author[2]{Tom\'{a}\v{s} Masa\v{r}\'{i}k}
\author[2]{Marcin Pilipczuk}
\author[3]{Roohani Sharma}
\author[4]{Magnus Wahlstr\"{o}m}
\affil[1]{Universit\'{e} Paris-Dauphine, PSL Research University, CNRS, UMR 7243, LAMSADE, 75016, Paris, France.}
\affil[2]{University of Warsaw, Warsaw, Poland}
\affil[3]{Max Planck Institute for Informatics, Saarland Informatics Campus, Saarbr\"{u}cken, Germany.}
\affil[4]{Royal Holloway, University of London, TW20 0EX, UK}

\begin{document}

\maketitle

\begin{textblock}{20}(0, 13.0)
\includegraphics[width=40px]{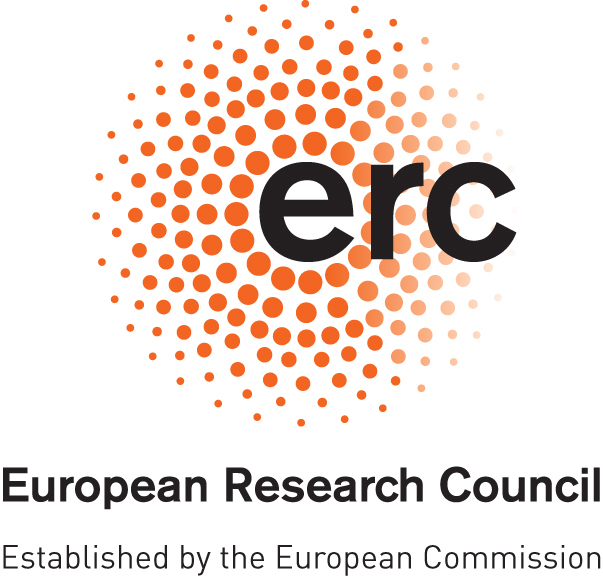}%
\end{textblock}
\begin{textblock}{20}(0, 13.9)
\includegraphics[width=40px]{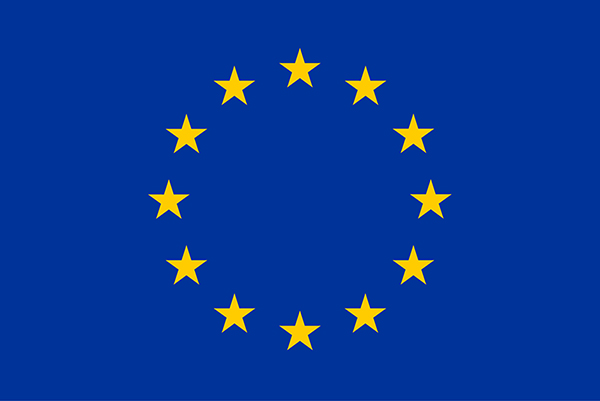}%
\end{textblock}

\begin{abstract}
One of the first application of
the recently introduced technique of \emph{flow-augmentation} [Kim et al., STOC 2022]
is a fixed-parameter algorithm for the weighted version of \textsc{Directed Feedback Vertex Set},
a landmark problem in parameterized complexity. 
In this note we explore applicability of flow-augmentation to other weighted graph separation problems parameterized by the size of the cutset.
We show the following.
\begin{itemize}
\item In weighted undirected graphs \textsc{Multicut} is FPT, both in the edge- and vertex-deletion version. 
\item The weighted version of \textsc{Group Feedback Vertex Set} is FPT, even with an oracle access to group operations.
\item The weighted version of \textsc{Directed Subset Feedback Vertex Set} is FPT.
\end{itemize}
Our study reveals \textsc{Directed Symmetric Multicut} as the next important graph separation problem whose parameterized complexity remains unknown, even in the unweighted setting.
\end{abstract}

\section{Introduction}

The family of graph separation problems includes a wide range of combinatorial problems where the goal is to remove a small part of the input graph to obtain
some separation properties. For example, in the \textsc{Multicut} problem, the input graph $G$ is equipped with a set of terminal pairs $\terms \subseteq V(G) \times V(G)$
and the separation objective is to destroy, for every $(s,t) \in \terms$, all paths from $s$ to $t$. 
In the \textsc{Subset Feedback Edge/Vertex Set} problems, the input graph $G$ is equipped with a set $R \subseteq E(G)$ of \emph{red edges} and the goal
is to destroy all cycles that contain at least one red edge.%
\footnote{In the literature, sometimes one considers red vertices instead of red edges. Since there are simple reductions between the variants (cf.~\cite{sfvs}), 
  we prefer to work with red edges.}
We remark that in directed graphs, one can equivalently require to destroy all closed walks
containing at least one red edge.

Both these problems (and many others) can be considered in multiple variants: graphs can be undirected or directed, we are allowed to delete edges or vertices, weights can be present, etc. 
In this paper we consider both edge- and vertex-deletion variants and both cardinality and weight budget for the solution. That is, the input graph $G$ is equipped with a weight function
$\weight$ that assigns positive integral weights to deletable objects (i.e., edges or vertices), and we are given two integers: $k$, the maximum number of deleted objects, and $W$, 
the maximum total weight of the deleted objects in the sought solution.

The study of parameterized complexity of graph separation problems has been a vivid line for the past two decades, and resulted in many tractability results and a wide range
of algorithmic techniques:
important separators and shadow removal~\cite{dfvs,dsfvs,dir-mwc,sfvs,marx:1,multicut-in-DAGs,marx:multicut,a2sat},
branching guilded by an LP relaxation~\cite{mwc-a-lp,sylvain,IwataWY16},
matroid-based techniques~\cite{ms1,ms2},
treewidth reduction~\cite{tw-red},
randomized contractions~\cite{rand-contr,lean-decomp},
and, most recent, flow-augmentation~\cite{ufl,dfl}.
However, the vast majority of these works considered only the unweighted versions of the problems,
for a very simple reason: we did not know how to handle their weighted counterparts.
In particular, one of the most fundamental notion --- important separators, introduced by Marx in 2004~\cite{marx:1} --- relies on a greedy argument that breaks down in the presence of weights.
The quest to understand the weighted counterparts of studied graph separation problems, with a specific goal to resolve the parameterized complexity of
the weighted version of \textsc{Directed Feedback Vertex Set} --- the landmark problem in parameterized complexity~\cite{dsfvs} --- was raised by Saurabh in 2017~\cite{rapcslides}
(see also~\cite{LokshtanovRS16}).

This question has been resolved recently by Kim et al.~\cite{dfl} with a new algorithmic technique called \emph{flow-augmentation}. 
Apart from proving fixed-parameter tractability of the weighted version of \textsc{Directed Feedback Vertex Set}, 
they also showed fixed-parameter tractability of \textsc{Chain SAT}, resolving another long-standing open problem~\cite{ChitnisEM17}. 
Both the aforementioned results are in fact the same relatively simple algorithm for a more general problem \textsc{Weighted Bundled Cut with Order},
     and solve also the weighted version of \textsc{Chain SAT}.

Very recently, Galby et al.~\cite{GalbyMSST22WG} used the flow-augmentation technique to design an FPT algorithm for weighted \textsc{Multicut} on trees. Our results thus extend theirs, by generalizing the input graphs from trees to arbitrary undirected graphs. 
     
\paragraph{Our results.}
The goal of this note is to explore for which other graph separation problems the flow-augmentation technique helps in getting fixed-parameter algorithms for weighted
graph separation problems. (All algorithms below are randomized; all randomization comes from the flow-augmentation technique.)

We start with the \textsc{Multicut} problem in undirected graphs, whose parameterized complexity --- in the unweighted setting --- had been a long-standing open problem until
being settled in the affirmative by two independent groups of researchers in 2011~\cite{thomasse:multicut,marx:multicut}.
\begin{theorem}\label{thm:multicut}
\textsc{Weighted Multicut}, parameterized by the cardinality of the cutset, is randomized FPT, both in the edge- and vertex-deletion variants.
\end{theorem}
Theorem~\ref{thm:multicut} follows from a combination of two arguments. 
First, we revisit the reduction of Marx and Razgon from \textsc{Multicut} to a \emph{bipedal variant}, presented in the conference version of their paper~\cite{marx:multicut-stoc}
and show how to replace one greedy step based on important separators with a different, weights-resilient step. 
Then, a folklore reduction to a graph separation problem called \textsc{Coupled Min-Cut}, spelled out in~\cite{ufl-arxiv}, does the job: 
the fixed-parameter tractability of a wide generalization of \textsc{Coupled Min-Cut}, including its weighted variant, is one of the main applications of flow-augmentation~\cite{ufl,dfl,dfl-csp}. 

\textsc{Multiway Cut} is a special case of \textsc{Multicut} where the input graph $G$ is equipped with a set $T \subseteq V(G)$ of terminals
and $\terms = \{(s,t)~|~s,t \in T, s \neq t\}$, that is, we are to destroy all paths between distinct terminals.
Thus, Theorem~\ref{thm:multicut} implies the following.
\begin{corollary}\label{cor:multiway-cut}
\textsc{Weighted Multiway Cut}, parameterized by the cardinality of the cutset, is randomized FPT, both in the edge- and vertex-deletion variants.
\end{corollary}

We remark that in directed graphs
the parameterized complexity of \textsc{Multicut} is fully understood:
without weights, it is W[1]-hard for 4 terminal pairs~\cite{PilipczukW18a} and FPT for 3 terminal pairs~\cite{dir-3-multicut},
but with weights it is already W[1]-hard for 2 terminal pairs~\cite{dir-3-multicut}, while for 1 terminal pair
it is known under the name of \textsc{Bi-Objective $st$-cut} and its fixed-parameter tractability follows easily via flow-augmentation~\cite{dfl}.
Furthermore, while \textsc{Multiway Cut} on directed graphs is FPT in the unweighted setting~\cite{dir-mwc},
on directed graphs \textsc{Multicut} with 2 terminal pairs reduces to \textsc{Multiway Cut} with two terminals~\cite{dir-mwc},
hence \textsc{Multiway Cut} with weights is W[1]-hard and without weights is FPT on directed graphs.

Then we turn our attention to \textsc{Group Feedback Edge/Vertex Set}. Here, the input graph $G$ is equipped with a group $\Gamma$, not necessarily Abelian,
and an assignment $\psi$, called the {\em group labels}, that assigns to every $e \in E(G)$ and $v \in e$ an element $\psi(e,v) \in \Gamma$ such that for $e=uv$
we have $\psi(e,u) + \psi(e,v) = 0$.%
\footnote{Thorough this paper, we use $+$ for the group operation, $0$ for the neutral element in the group, and $-$ for the group inverse,
  to conform with the standard terminology of \emph{null cycles} in GFVS. We note that this is in tension with the convention that a group operation written
  as $+$ tends to imply an Abelian group.} 
With a walk $C = (v_1,e_1,v_2,e_2,\ldots,v_\ell,e_\ell,v_{\ell+1})$ we associate a sum $\psi(C) = \sum_{i=1}^\ell \psi(e_i,v_i)$;
a walk $C$ is a \emph{null walk}
in $(G,\psi)$
 if $\psi(C) = 0$ and \emph{non-null} otherwise. This is well-defined even for non-Abelian groups, i.e.,
a cycle being null or non-null does not depend on the direction of traversal or the choice of starting vertex $v_1$~\cite{CyganPP16}.
The separation goal is to destroy all non-null cycles (equivalently, all non-null closed walks)
by edge or vertex deletions.

\begin{theorem}\label{thm:gfvs}
\textsc{Weighted Group Feedback Edge Set}
and 
\textsc{Weighted Group Feedback Vertex Set},
parameterized by the cardinality of the cutset, are randomized FPT.
\end{theorem}
Since \textsc{Weighted Subset Feedback Edge/Vertex Set} can be modeled as 
\textsc{Weigted Group Feedback Edge/Vertex Set} with group
$\Gamma = \mathbb{Z}_2^R$ (cf.~\cite{CyganPP16}), we immediately have the following corollary.
\begin{corollary}\label{cor:sfvs}
\textsc{Weighted Subset Feedback Edge Set}
and 
\textsc{Weighted Subset Feedback Vertex Set},
parameterized by the cardinality of the cutset, are randomized FPT.
\end{corollary}

The currently fastest FPT algorithm for (unweighted) \textsc{Group Feedback Vertex Set} is due to Iwata, Wahlstr\"{o}m, and Yoshida~\cite{IwataWY16}
and uses sophisticated branching guided by an LP relaxation. 
To prove Theorem~\ref{thm:gfvs},
we revisit an older (and less efficient) FPT algorithm due to Cygan et al.~\cite{CyganPP16} that performs some branching steps to reduce the problem
to multiple instances of \textsc{Multiway Cut}.
We observe that the branching easily adapts to the weighted setting, and 
the algorithm for \textsc{Weighted Multiway Cut} is provided by Corollary~\ref{cor:multiway-cut}.

We now move to directed graphs. As already mentioned, the parameterized complexity of both weighted and unweighted \textsc{Directed Multicut} (and \textsc{Directed Multiway Cut}) is already fully understood~\cite{dir-3-multicut,PilipczukW18a}.
Our main result here is fixed-parameter tractability of \textsc{Weighted Directed Subset Feedback Edge/Vertex Set}.
\begin{theorem}\label{thm:dsfvs}
\textsc{Weighted Directed Subset Feedback Edge Set}
and 
\textsc{Weighted Directed Subset Feedback Vertex Set},
parameterized by the cardinality of the cutset, are randomized FPT.
\end{theorem}
Theorem~\ref{thm:dsfvs} follows from a surprisingly delicate reduction to \textsc{Weighted Bundled Cut with Order}, known to be FPT via flow-augmentation~\cite{dfl}.

\textsc{Skew Multicut} is a special case of \textsc{Directed Multicut} where the set $\terms$ has the form $\{(s_i,t_j)~|~1 \leq i \leq j \leq \ell\}$ for some terminals
$s_1,\ldots,s_\ell,t_1,\ldots,t_\ell \in V(G)$. 
\textsc{Skew Multicut} naturally arises in the context of \textsc{Directed Feedback Vertex Set} if one applies the iterative compression technique. 
In the unweighted setting, \textsc{Skew Multicut} is long known to be FPT parameterized by the size of the cutset~\cite{dfvs}.
With weights, \cite{dfl} showed that \textsc{Skew Multicut} is FPT when parameterized by $k + \ell$. 
We observe a simple reduction to \textsc{Weighted Directed Subset Feedback Vertex Set}, yielding fixed-parameter tractability when parameterizing by $k$ only.

\begin{corollary}\label{cor:skew-multicut}
\textsc{Weighted Skew Multicut}, parameterized by the cardinality of the cutset, is randomized FPT, both in the edge- and vertex-deletion variants.
\end{corollary}
\begin{proof}
Let $(G,(s_i,t_i)_{i=1}^\ell,\weight,k,W)$ be a \textsc{Weighted Skew Multicut} instance (in the edge- or vertex-deletion setting) where $\weight$ is the weight function (on the edges or vertices respectively) and $W$ is the weight budget of the solution.
Construct a graph $G'$ and a set of red edges $R$ as follows: start with $G'=G$, $R =\emptyset$ and, for every $1 \leq i \leq j \leq \ell$, introduce a red edge $(t_j,s_i)$
and add it to $G'$ (in the edge-deletion setting, the new edge has weight $W+1$, that is, it is effectively undeletable).
It is easy to see that the resulting \textsc{Weighted Directed Subset Feedback Edge/Vertex Set} instance $(G',R,\weight,k,W)$ is equivalent to the input
\textsc{Weighted Skew Multicut} instance: any closed walk in $G'$ involving a red edge contains a subpath from $s_i$ to $t_j$ for some $1 \leq i \leq j \leq \ell$ without any red edge,
  and any path in $G$ from $s_i$ to $t_j$ for $1 \leq i \leq j \leq \ell$ closes up to a cycle with a red edge $(t_j,s_i)$ in $G'$.
\end{proof}

The running time bounds of all our algorithms are of the form $2^{\mathrm{poly}(k)} \mathrm{poly}(|V(G)|)$, where both polynomial dependencies have unspecified large degree coming from the
use of involved flow-augmentation-based algorithms of~\cite{dfl,dfl-csp}.

\paragraph{Organization.}
We introduce the necessary tools, in particular the used corollaries of the flow-augmentation technique, in Section~\ref{sec:prelims}.
Theorem~\ref{thm:multicut} is proven in Section~\ref{sec:multicut},
Theorem~\ref{thm:gfvs} is proven in Section~\ref{sec:gfvs},
and Theorem~\ref{thm:dsfvs} is proven in Section~\ref{sec:dsfvs}.
Section~\ref{sec:conc} concludes the paper and identifies \textsc{Directed Symmetric Multicut} as a next problem whose parameterized complexity remains open.

\section{Preliminaries}\label{sec:prelims}

\subsection{Edge- and vertex-deletion variants}

In directed graphs, there is a simple reduction from the vertex-deletion setting 
to the edge-deletion one: replace every vertex $v$ with two vertices $v^+$ and $v^-$
and an edge $(v^-,v^+)$; every previous arc $(u,v)$ becomes an arc $(u^+,v^-)$.
Now, the deletion of the vertex $v$ corresponds to the deletion of the arc $(v^-,v^+)$.
Hence, in Section~\ref{sec:dsfvs} we will consider only the edge-deletion variant,
that is, \textsc{Directed Subset Feedback Edge Set}.

No such simple reduction is available in undirected graphs and, in fact, 
in some cases the vertex-deletion variant turns out to be significantly more difficult
(cf. the \textsc{$k$-Way Cut} problem~\cite{rand-contr,KawarabayashiT11,marx:1}). 
In the presence of weights, there is a simple reduction from the edge-deletion variant
to the vertex-deletion variant: subdivide every edge with a new vertex that inherits the weight of the edge
it is placed on, and set the weight of the original vertices to $+\infty$, making them undeletable. 
(For clarity, we allow the weight function $\weight$ to attain the value $+\infty$, which is equivalent
to any weight larger than $W$ and models an undeletable edge or vertex.)
Thus, both in Section~\ref{sec:multicut} and in Section~\ref{sec:gfvs} we consider the vertex-deletion variants. 

\subsection{Iterative Compression}

All problems considered in this paper are monotone in the sense that deletion of an edge or a vertex from the input graph
cannot turn a \yes-instance into a \no-instance. 
This allows to use the standard technique of \emph{iterative compression}~\cite{ReedSV04}:
We enumerate $V(G) = \{v_1,v_2,\ldots,v_n\}$ for $n = |V(G)|$, denote $G_i = G[\{v_1,\ldots,v_i\}]$ for $0 \leq i \leq n$
and iteratively solve the problem on graphs $G_0$, $G_1$, \ldots, $G_n = G$. 
If the instance for $G_i$ turns out to be a \no-instance, we deduce that the input instance is a \no-instance, too.
Otherwise, the computed solution for $G_i$ allows us to infer a set $X' \subseteq V(G_i)$ of size at most $k$
such that in $G_i-X'$ already has the desired separation (i.e., induces a \yes-instance with parameter $k=0$).
We set $X = X' \cup \{v_{i+1}\}$ and observe that $G_{i+1}-X = G_i-X'$ and $|X| \leq k+1$.

Furthermore, in all considered problems, using self-reducibility
it is immediate to turn an algorithm that
only gives a yes/no answer into an algorithm that, in case of a positive answer, returns
a cutset that is a solution.

Hence, in all our algorithmic results, we can solve a compression version of the problem.
That is, we can assume that our algorithm is additionally given on input a set $X \subseteq V(G)$ of size at most $k+1$ such that $G-X$
already satisfies the desired separation (i.e., has no cycle with a red edge in case of \textsc{Subset Feedback Edge Set} etc.). 

Furthermore, in the problems that involve vertex deletions (i.e., Sections~\ref{sec:multicut} and~\ref{sec:gfvs}), we can additionally branch on the set $X$
into $2^{|X|}$ options, guessing a set $Y \subseteq X$ of vertices that are included in the sought solution. 
In each branch, we delete $Y$ from the graph and the set $X$, decrease $k$ by $|Y|$ and decrease $W$ by the weight of $Y$. 
Furthemore, we set the weight of the remaining vertices of $X \setminus Y$ to $+\infty$, so they become undeletable. 
In other words, in Sections~\ref{sec:multicut} and~\ref{sec:gfvs} we solve a disjoint compression variant of the problem, where the sought solution is 
supposed to be disjoint with the set $X$.

\subsection{Generalized Digraph Pair Cut}\label{prelims:gdpc}

We will not need flow-augmentation in its raw form, but only one algorithmic corollary
of this technique.

An instance of \textsc{Generalized Digraph Pair Cut} (\gdpc{} for short) consists of:
\begin{itemize}
\item a directed multigraph $G$ with two distinguished vertices $s,t \in V(G)$;
\item a multiset $\pairs$ of (unordered) pairs of vertices of $G$, called \emph{clauses};
\item a family $\bundles$ of pairwise disjoint subsets of $E(G) \cup \pairs$
called \emph{bundles} such that no bundle contains two copies of the same arc or two
copies of the same pair;
\item a weight function $\weight : \bundles \to \mathbb{Z}_+$;
\item two integers $k$ and $W$.
\end{itemize}
A set $Z \subseteq E(G)$ is a \emph{cut} in a \gdpc{} instance
$\inst = (G,s,t,\pairs,\bundles,\weight,k,W)$ if
$Z \subseteq E(G) \cap \bigcup_{B \in \bundles} B$ (i.e., $Z$ contains only edges of bundles)
and there is no path from $s$ to $t$ in $G-Z$. 
A cut $Z$ \emph{violates} an edge $e \in E(G)$ if $e \in Z$ and
\emph{violates} a clause $uv \in \pairs$ if both $u$ and $v$ are reachable from $s$ in $G-Z$.
A bundle is \emph{violated} by $Z$ if it contains an edge or a clause violated by $Z$.
An edge, a clause, or a bundle not violated by $Z$ is \emph{satisfied} by $Z$.
A cut $Z$ is a \emph{solution} if every clause violated by $Z$ is part of a bundle,
$Z$ violates at most $k$ bundles, and the total weight
of violated bundles is at most $W$. (Recall that a cut is required to contain only edges of bundles, that is, it satisfies all edges outside bundles.)
The \gdpc{} problem asks for an existence of a solution.

\gdpc{}, parameterized by $k$, is W[1]-hard even in the unweighted setting
and without clauses: it suffices to have bundles consisting of two edges for the hardness~\cite{MarxR09}.
However, flow-augmentation yields fixed-parameter tractability of some specific useful restrictions
of \gdpc{}. 

For a bundle $B \in \bundles$, let $V(B)$ be the set of vertices that are involved 
in an arc or a clause of $B$
and let $G_B$ be an undirected graph with $V(G_B) = V(B) \setminus \{s,t\}$ and $uv \in E(G_B)$
if $B$ contains an arc $(u,v)$, an arc $(v,u)$, or a clause $uv$.
A bundle $B$ is $2K_2$-free if $G_B$ is $2K_2$-free, that is, it does not contain $2K_2$
(the four-vertex graph consisting of two independent edges) as an induced subgraph. 
An instance $\inst$ of \gdpc{} is $2K_2$-free if every bundle of $\inst$
is $2K_2$-free.
Finally, an instance $\inst$ is $\maxarity$-bounded if for every $B \in \bundles$
we have $|V(B)| \leq \maxarity$.

One of the main algorithmic corollaries of the flow-augmentation technique is 
the tractability of $2K_2$-free $\maxarity$-bounded instances of \gdpc{}.

\begin{theorem}[\cite{dfl-csp}, Theorem~3.3]\label{thm:gdpc:2k2}
There exists a randomized polynomial-time algorithm for \textsc{Generalized Digraph Pair Cut}
restricted to $2K_2$-free $\maxarity$-bounded instances that never accepts a \no-instance
and accepts a \yes-instance with probability $2^{-\mathrm{poly}(k,\maxarity)}$. 
\end{theorem}

For \textsc{Directed Subset Feedback Edge Set} it will be more convenient to look at a
different restriction of \gdpc{}.
Let $\inst = (G,s,t,\emptyset,\bundles,\weight,k,W)$ be a \gdpc{}
instance without clauses.
An arc $e \in E(G)$ is \emph{crisp} if it is not contained in any bundle, and \emph{soft}
otherwise.
An arc $e \in E(G)$ is \emph{deletable} if it is soft and there is no copy of $e$ in
$G$ that is crisp.
Note that a cut needs to contain soft arcs only and in fact we can restrict our attention
to cuts containing only deletable arcs. 
A bundle $B \in \bundles$ has \emph{pairwise linked deletable edges} if for every two deletable
arcs $e_1,e_2 \in B$ that are not incident with either $s$ or $t$, there is a path from
an endpoint of one of the edges to an endpoint of the other that does not use
an edge of another bundle (i.e., uses only edges of $B$ and crisp edges).

In~\cite{dfl}, a notion of \textsc{Bundled Cut with Order} has been introduced
as one variant of \gdpc{} without clauses that is tractable.
In~\cite{dfl-csp}, it was observed that the notion of \emph{pairwise linked deletable edges}
is slightly more general than the ``with order'' assumption and is more handy. 
\begin{theorem}[\cite{dfl-csp}, Theorem~3.21]\label{thm:gdpc:linked}
There exists a randomized polynomial-time algorithm that, given a
\gdpc{} instance $\inst = (G,s,t,\emptyset,\bundles,\weight,k,W)$ with no clauses
and whose every bundle has pairwise linked deletable edges, never accepts a \no-instance
and accepts a \yes-instance with probability $2^{-\Oh(k^4 d^4 \log(kd))}$ where
$d$ is the maximum number of deletable arcs in a single bundle.
\end{theorem}
Note that if $\inst$ is $\maxarity$-bounded, then $d \leq \maxarity^2$.

\section{Multicut}\label{sec:multicut}

This section is devoted to the proof of Theorem~\ref{thm:multicut}.

As discussed in Section~\ref{sec:prelims}, we can restrict ourselves to the vertex-deletion
variant. 
 Let $\inst = (G, \terms, \weight, k,W)$ be an instance of \textsc{Weighted Multicut}.
Let $T = \bigcup_{(s,t) \in \terms} \{s,t\}$ be the set of all terminals.
By a simple reduction, we can assume that all terminals have weight $+\infty$
and form an independent set: 
for every $(s,t) \in \terms$, add a new vertex $s'$ adjacent to $s$,
add a new vertex $t'$ adjacent to $t$, set $\weight(s') = \weight(t') = +\infty$
and replace $(s,t)$ with $(s',t')$ in $\terms$. 

We also use iterative compression, but in the ordering $v_1,\ldots,v_n$ of $V(G)$
we start with terminals. Note that the subgraph of $G$ induced by the terminals is edgeless
and thus admits a solution being the emptyset.
As a result, using standard iterative compression step discussed in Section~\ref{sec:prelims}
we can assume that the algorithm is given access to a set
$X \subseteq V(G) \setminus T$ of size $|X| \leq k+1$ 
such that for every $(s,t) \in \terms$ there is no path from $s$ to $t$ in $G-X$
and we are to check if there is a solution disjoint with $X$.
We can set $\weight(x) = +\infty$ for every $x \in X$. 

We closely follow the steps in Section~5 of~\cite{marx:multicut:stoc}, reengineering only
one branching step that originally uses important separators.

Fix a hypothetical solution $Z$. We first guess how the vertices of $X$ are partitioned
between connected components of $G-Z$. This results in $2^{\Oh(k \log k)}$ subcases.
If two vertices of $X$ are guessed to be in the same connected component of $G-Z$, we can 
merge them into a single vertex (recall that the solution $Z$ is disjoint with $X$). 
After this step, we can assume that every connected component of $G-Z$ contains at most
one vertex of $X$
and $X$ is an independent set. 
For brevity, we say that $Y \subseteq V(G) \setminus (X \cup T)$ is a \emph{multiway cut}
if every connected component of $G-Y$ contains at most one vertex of $X$.
Thus, it suffices to develop a randomized FPT algorithm that
(a) accepts with constant probability an instance that admits a solution that is a multiway cut;
(b) never accepts a \no-instance.

An instance is \emph{bipedal} if $X$ is an independent set and for every connected component
$C$ of $G-X$, we have $|N_G(C)| \leq 2$, that is, $C$ is adjacent to at most two vertices
of $X$. 
In Section~\ref{ss:multicut:bipedal} we show how to reduce a bipedal instance
to a \gdpc{} instance handled by Theorem~\ref{thm:gdpc:2k2}. 
We emphasize that we do not claim authorship of this reduction:
while there is no citeable source of this reduction, it has been floating around
in the community in the last years. 
The reduction, in the edge-deletion setting (and leading to an undirected analog
 of \gdpc{}) has been spelled out in~\cite{ufl-arxiv}. 
We include it here for completeness of the argument.

Section~\ref{ss:multicut:branch} describes a branching algorithm, closely following
the arguments of~\cite{marx:multicut:stoc}, whose goal is to break
connected components $C$ of $G-X$ with $|N_G(C)| > 2$. In the leaves
of the branching process we obtain bipedal instances that are passed
to the algorithm of Section~\ref{ss:multicut:bipedal}. 

\subsection{Branching on a multilegged component}\label{ss:multicut:branch}

The algorithm is a recursive branching routine
on an instance $(G,\terms,\weight,k,W,X)$ 
where $X$ is an independent set and a multicut for $\terms$,
and the hypothetical solution is also a multiway cut for $X$.
In the beginning $|X| \leq k+1$ as discussed earlier. During the branching algorithm one may delete vertices, merge vertices or grow the set $X$ while maintaining that the hypothetical solution is also a multiway cut for (the new) $X$.

In a recursive call, we start
with a few cleaning steps. At every moment, apply the first applicable reduction step.

\begin{enumerate}
\item \label{rr:1} If $\emptyset$ is a solution, return \yes.
\item If $k \leq 0$, $W \leq 0$, or $X$ is not an independent set, return \no.
\item If the number of connected components $C$ of $G-X$ with $|N_G(C)| > 1$ is more than
$k$, return \no. (Note that every such component needs to contain at least one vertex
    of every multiway cut.)
\item \label{rr:4} If there exists $x \in X$ such that the cardinality of the minimum-cardinality
 vertex cut between $x$ and $X \setminus x$ is of size larger than $k$, return \no.
 (Recall that the solution is also a multiway cut for $X$.)
\item \label{rr:5} If there exists a vertex $v$ that admits a family $\mathcal{P}$ of $k+2$ paths
  that start in $v$, end in distinct vertices of $X$, and are vertex-disjoint except for $v$,
  delete $v$, decrease $k$ by one, decrease $W$ by $\weight(v)$, and recurse.
  (Note that every such vertex $v$ needs to be included in any multiway cut
   of size at most $k$.)
\item   \label{rr:6}
  If there exists a connected component $C$ of $G$
 that do not contain both vertices of any terminal pair $(s,t) \in \terms$ and contains at most one vertex of $X$, 
  delete it and all terminal pairs involving a vertex of $C$.
  (Recall that for every $(s,t) \in \terms$, the terminals $s$ and $t$ lie in different
   connected components of $G-X$. Hence, this rule applies to any component $C$ that contains no vertex of $X$
   and to any isolated vertex of $X$.)

     \item \label{rr:size-of-X} 
   If $|X| > k(k+1)$, return \no.  (Since the previous reduction rule is inapplicable, for every multiway cut $Z$,
  every $z \in Z$ is adjacent to at most $k+1$ connected components of $G-Z$ that contain
  a vertex of $X$. Also, since $Z$ is a multiway cut for $X$, every vertex of $X$ is in a distinct connected component of $G-Z$.
Further, since the previous rule is not applicable, there does not exist a connected component of $G-Z$ that has no neighbour in $Z$. Indeed, 
as such an isolated component will have at most one vertex from each terminal pair in $\terms$ because $Z$ is a solution and at most one vertex of $X$ since $Z$ is a multiway cut of $Z$.
Therefore, $|X|$ is at most the number of connected components of $G-Z$ that intersect $X$, which is upper bounded by $|Z| (k+1) \leq k(k+1)$.)
\item If the current instance is bipedal, pass it to the algorithm of Section~\ref{ss:multicut:bipedal}. 
\end{enumerate}

A component $C$ of $G-X$ is \emph{nontrivial} if $|N_G(C)| > 1$. 
If neither of the reduction steps is applicable, we have at most $k$ nontrivial connected components and the size of the neighbourhood of each component of $G-X$ is at most $|X| \leq k(k+1)$.

At every branching step, we will ensure that one of the following progresses happen
in any recursive call:
\begin{itemize}
\item the instance will be resolved immediately by reduction rules~\ref{rr:1}-\ref{rr:4} or~\ref{rr:size-of-X}, or 
\item the parameter $k$ decreases, or
\item the parameter $k$ stays the same, but the number of nontrivial connected components plus the number vertices of $X$ adjacent to a nontrivial component increases.
\end{itemize}
We observe that the reduction rules do not reverse the above progress.
That is, Rule~\ref{rr:5} can decrease the number of nontrivial connected components or the number of vertices of $X$ incident with a nontrivial connected component, but at the same time decreases $k$ by one,
while Rule~\ref{rr:6} cannot delete a nontrivial connected component.

After the application of the described reduction rules, the number of non-trivial components is at most $k$ and the size of $X$ is at most $k(k+1)$. Thus, the depth of the recursion is bounded by $\Oh(k^3)$. 

Let $C$ be a component of $G-X$ with $|N_G(C)| > 2$. (It exists as the instance is not bipedal.)
For a subset $B \subseteq C$ and a function $f : B \to N_G(C)$,
we construct an instance $\inst_f$ as follows: for every $v \in B$, we merge $v$
onto the vertex $f(v)$ (we use $f(v)$ as the name of the resulting vertex and the
resulting vertex still belongs to $X$).
We say that $B$ is a \emph{shattering set} if for every $f : B \to N_G(C)$, the instance
$\inst_f$ either contains strictly more nontrivial components than the current instance,
or recursing on $\inst_f$ will result in returning an immediate answer by one
of the first four reduction rules. 

The main technical contribution of Section~5 of~\cite{marx:multicut:stoc}
is the following statement.
\begin{lemma}\label{lem:multicut:shatter}
Given an instance $(G,\terms,\weight,k,W)$ together with a set $X \subseteq V(G) \setminus T$ such that
in $G-X$ there is no path from $s$ to $t$ for any $(s,t) \in \terms$,
   and a component $C$ of $G-X$ with $|N_G(C)| > 2$, one can find a shattering set $B \subseteq C$
   of size at most $3k$ in polynomial time.
\end{lemma}

We apply Lemma~\ref{lem:multicut:shatter} to $C$, obtaining a set $B$ of size at most $3k$.
We branch, guessing the first of the following options that happens
with regards to a hypothetical solution $Z$:
\begin{enumerate}
\item There is a vertex $v \in B \cap Z$. We guess $v$,
  delete $v$ from the graph, decrease $k$ by one,
  decrease $W$ by $\weight(v)$, and recurse. 
  This gives $|B| \leq 3k$ subcases and in each subcase $k$ drops.
\item For every $v \in B$, the connected component of $G-Z$ that contains $v$ also contains a vertex
of $X$. 
For every $v \in B$, we guess a vertex $f(v) \in N_G(C)$ that is in the same connected component
of $G-Z$ as $v$. As $|N_G(C)| \leq |X| \leq k(k+1)$ and $|B| \leq 3k$, there are $2^{\Oh(k \log k)}$
options for $f : B \to N_G(C)$.
We recurse on $\inst_f$.
To see that we obtain progress, observe that:
\begin{itemize}
\item the parameter $k$ stays the same;
\item if $X$ is not an independent set, the recursive call returns \no{} immediately;
\item otherwise,  
  the fact that $B$ is a shattering set implies that in each instance $\inst_f$ the number of non-trivial components increases,
  while the connectivity of $C$ implies that every vertex of $N_G(C)$ remains adjacent to a nontrivial connected component, 
  so the set of vertices of $X$ adjacent to a nontrivial connected component does not change. 
\end{itemize}
\item There exists $v \in B$ such that the connected component of $G-Z$ that contains $v$
is disjoint with $X$.
Here, \cite{marx:multicut:stoc} branches on an important separator separating $v$ from $X$. 
This does not work in the presence of weights, so we need to proceed differently.
We insert $v$ into $X$, set its weight to $+\infty$, and recurse. 
Clearly, the hypothetical solution $Z$ remains a solution and, if the guess is correct,
$Z$ remains a multiway cut (with regards to the enlarged set $X$).
To see that we obtain progress, observe that:
\begin{itemize}
\item the parameter $k$ stays the same;
\item if $X$ is not an independent set, the recursive call returns \no{} immediately;
\item otherwise, first observe that in the right guess $v$ has no neighbors in the set $X$; therefore, for every $y \in N_G(C)$, there exists a connected component $C_y$ of $C-\{v\}$
with $y \in N_G(C_y)$ and as $v \in N_G(C_y)$ due to connectivity of $C$, $C_y$ is a new nontrivial component;
hence the number of vertices of $X$ that are incident with a nontrivial connected component
increases as both $v$ and the whole $N_G(C)$ are now adjacent to nontrivial connected components; 
furthermore, the number of nontrivial connected components does not decrease as at least one new nontrivial component is created in the place of $C$ since $N_G(C) \neq \emptyset$.
\end{itemize}
\end{enumerate}

Hence, the recursive step invokes $2^{\Oh(k \log k)}$ recursive subcalls, in each
obtaining the promised progress. Every single recursive call takes polynomial time.
Consequently, the branching algorithm takes $2^{\Oh(k^4\log k)} n^{\Oh(1)}$ time
and results in $2^{\Oh(k^4 \log k)}$ leaves of the recursion trees
that give either an immediate answer or a bipedal instance, which is passed to Section~\ref{ss:multicut:bipedal}. 

\subsection{Solving a bipedal instance}\label{ss:multicut:bipedal}

We now show how to reduce a bipedal instance to a \gdpc{} instance
where every bundle consists of at most two arcs and a single clause containing the heads
of these two arcs. These bundles are $2K_2$-free and $4$-bounded and hence can be solved by Theorem~\ref{thm:gdpc:2k2}
in randomized FPT time $2^{k^{\Oh(1)}} n^{\Oh(1)}$. 
This is essentially repeating the arguments of Lemma~7.1 of~\cite{ufl-arxiv},
adjusted for the vertex-deletion setting and \gdpc{}.

We start with a graph $H$ consisting of vertices $s$ and $t$.
For every component $C$ of $G-X$, proceed as follows.
Recall that $|N_G(C)| \in \{1,2\}$. 
Denote one of the elements of $N_G(C)$ as $s_C$ and the other as $t_C$, if present. 
For every $v \in V(G)$, create four vertices $v_s^+$, $v_s^-$, $v_t^+$, $v_t^-$,
arcs $(v_s^-,v_s^+)$, $(v_t^-, v_t^+)$, and a clause $v_s^+v_t^+$.
The two constructed arcs and the constructed clause form a bundle $B_v$
of weight $\weight(v)$. 
These are all the bundles that we will construct; 
all subsequent arcs and clauses will not be in any bundle and thus will be undeletable. 
For every connected component $C$ of $G-X$ and $uv \in E(G[C])$, add arcs
$(u_s^+, v_s^-)$,
$(v_s^+, u_s^-)$
$(u_t^+, v_t^-)$, and
$(v_t^+, u_t^-)$.
For every $vs_C \in E(G)$ with $v \in C$, add arcs
$(s,v_s^-)$ and $(v_t^+,t)$.
For every $vt_C \in E(G)$ with $v \in C$, add arcs
$(s,v_t^-)$ and $(v_s^+,t)$.

Finally, for every $(u,v) \in \terms$ we proceed as follows. 
Note that $u$ and $v$ are in distinct connected components of $G-X$,
     say $C_u$ and $C_v$.
For every $x \in N_G(C_u) \cap N_G(C_v)$ we proceed as follows.
Say $x = \alpha_{C_u}$ and $x = \beta_{C_v}$
 for $\alpha,\beta \in \{s,t\}$. 
Add a clause $u^-_\alpha v^-_\beta$. 
This finishes the description of the \gdpc{} instance 
$\inst' = (H,s,t,\pairs,\bundles,\weight,k,W)$. 
It is immediate that the instance satisfies the prerequisities of 
Theorem~\ref{thm:gdpc:2k2} with $\maxarity=4$. 

It remains to check the equivalence of the instance $\inst'$ of \gdpc{}  with the input 
instance $\inst = (G,\terms,\weight,k,W)$ together with the set $X$. We do it in the next two lemmata,
completing the proof of Theorem~\ref{thm:multicut}. Recall $T$ is the set of all terminal vertices.

\begin{lemma}\label{lem:multicut:forward}
If $Z \subseteq V(G) \setminus (X \cup T)$ is a solution that is also a multiway cut for $X$,
then $Z' = \bigcup_{v \in Z} B_v \cap E(H)$ is a cut in $\inst'$
that satisfies all clauses outside $B_v$ for $v \in Z$. 
\end{lemma}
\begin{proof}
Assume first that $H-Z'$ contains a path $P'$ from $s$ to $t$.
Observe that there exists a component $C$ of $G-X$ and $\alpha \in \{s,t\}$
such that all internal vertices of $P'$ are of the form $v_\alpha^+$ or $v_\alpha^-$
for $v \in C$. 
Then, the path $P'$ induces a path from $s_C$ to $t_C$ via $C$ in $G-Z$, a contradiction
to the assumption that $Z$ is a multiway cut for $X$. 

Assume now that $Z'$ violates a clause $v_s^+v_t^+$ in $B_v$.
Then first observe that $v \in C$, for a component $C$ of $G-X$. 
Let $P_s'$ be a path from $s$ to $v_s^+$ in $H-Z'$
and let $P_t'$ be a path from $s$ to $v_t^+$ in $H-Z'$. 
In $G-Z$, the path $P_s'$ yields a path $P_s$ from $s_C$ to $v$ and 
the path $P_t'$ (reversed) yields a path $P_t$ from $v$ to $t_C$.
Together, $P_s$ and $P_t$ yield a path from $s_C$ to $t_C$ in $G-Z$, a contradiction
to the assumption that $Z$ is a multiway cut.

Finally, assume that $Z'$ violates a clause $u^-_\alpha v^-_\beta$ for some
$(u,v) \in \terms$, where $C_u$ and $C_v$ are the components of $G-X$ containing
$u$ and $v$, respectively, $x \in N_G(C_u) \cap N_G(C_v)$, and 
$x = \alpha_{C_u}$, $x = \beta_{C_v}$ for $\alpha,\beta \in \{s,t\}$. 
Let $P_u'$ be a path from $s$ to $u^-_\alpha$ in $H-Z'$ and 
let $P_v'$ be a path from $s$ to $v^-_\beta$ in $H-Z'$.
In $G-Z$, $P_u'$ yields a path $P_u$ from $u$ to $x = \alpha_{C_u}$
and $P_v'$ yields a path $P_v$ from $v$ to $x = \beta_{C_v}$. 
Together, $P_u$ and $P_v$ yield a path from $u$ to $v$ in $G-Z$, a contradiction
to the assumption that $Z$ is a solution.
\end{proof}

\begin{lemma}\label{lem:multicut:backward}
If $Z'$ is a cut in $\inst'$ that satisfies all clauses that are not in bundles
and $Z$ consists of those $v$ such that $Z'$
violates $B_v$, then $Z$ is a solution to $\inst$ that is also a multiway cut for $X$.
\end{lemma}
\begin{proof}
We first show that $Z$ is a multiway cut for $X$. 
By contradiction, assume that there exists a component $C$ of $G-X$ and a path $P$
from $s_C$ to $t_C$ via $C$ that avoids $Z$. Let $v$ be an arbitrary vertex
of $P$ in $C$. Then, the prefix of $P$ from $s_C$ to $v$ lifts to a path $P_s'$
in $H-Z'$ from $s$ to $v_s^+$. Similarly, the suffix of $P$ from $v$ to $t_C$, reversed,
lifts to a path $P_t'$ in $H-Z'$ from $s$ to $v_t^+$.
Hence, $Z'$ violates the clause $v_s^+v_t^+$ and hence the bundle $B_v$, which is a contradiction.

Consider now $(u,v) \in \terms$ and assume there is a path $P$ from $u$ to $v$ in $G-Z$.
Since $Z$ is a multiway cut for $X$, $P$ contains at most one vertex of $X$. 
Since $u$ and $v$ are in distinct connected components of $G-X$ (say, $C_u$ and $C_v$, respectively),
$P$ contains at least one vertex of $X$. 
That is, $P$ starts in $u$, continues via $C_u$ to a vertex $x \in N_G(C_u) \cap N_G(C_v)$,
and then continues via $C_v$ to $v$. 
The prefix of $P$ from $u$ to $x$ (reversed) lifts to a path $P_u'$ in $H-Z'$
from $s$ to $u_\alpha^-$ where $x = \alpha_{C_u}$, $\alpha \in \{s,t\}$.
The suffix of $P$ from $x$ to $v$ lifts to a path $P_v'$ in $H-Z'$
from $s$ to $v_\beta^-$ where $x = \beta_{C_v}$, $\beta \in \{s,t\}$. 
Hence, the clause $u_\alpha^- v_\beta^-$ is violated by $Z'$, a contradiction.
This finishes the proof of Lemma~\ref{lem:multicut:backward}. 
\end{proof}

With the discussion above, Lemmata~\ref{lem:multicut:forward} and~\ref{lem:multicut:backward} conclude the proof of Theorem~\ref{thm:multicut}.

\section{Group Feedback Edge/Vertex Set}\label{sec:gfvs}

This section is devoted to the proof of Theorem~\ref{thm:gfvs}. 
In fact, we just closely follow the arguments of~\cite{CyganPP16} and verify that
they work also in the weighted setting. The algorithm reduces the problem
to multiple instances of \textsc{Multiway Cut}. Here, in the presence of weights,
we apply the algorithm of Theorem~\ref{thm:multicut} to solve \textsc{Weighted Multiway Cut} (in particular, we use Corollary~\ref{cor:multiway-cut}).

As discussed in Section~\ref{sec:prelims}, we can focus on the vertex-deletion variant
\textsc{Group Feedback Vertex Set}
Using iterative compression (Section~\ref{sec:prelims}) we assume that,
apart from the input instance $(G,\psi,\weight,k,W)$, we are given
a set $X \subseteq V(G)$ of size at most $k+1$ such that $G-X$ has no non-null cycles
and the goal is to find a solution disjoint from $X$. We set $\weight(x) = +\infty$
for every $x \in X$. Recall that in this problem the input graph $G$ is equipped with a group $\Gamma$.

For a graph $H$ with group labels $\psi$, a \emph{consistent labeling} is
a function $\phi : V(H) \to \Gamma$ such that 
$\phi(v) = \phi(u) + \psi(e,u)$ for every $e = uv \in E(H)$.
It is easy to see that $(H,\psi)$ has no non-null cycle if and only
it admits a consistent labeling.

\paragraph{Untangling.}
By standard relabelling process, we can assume that $\psi(e,v) = 0$ for every
$e \in E(G-X)$ and $v \in e$; we call such an instance \emph{untangled}.
Since $G-X$ has no non-null cycles, there exists $\phi : V(G) \setminus X \to \Gamma$
such that for every $e=uv \in E(G-X)$ we have $\phi(v) = \phi(u) + \psi(e,u)$.
For every $e=uv \in E(G-X)$ we relabel $\psi(e,u) := \phi(u) + \psi(e,u) - \phi(v)$
and $\psi(e,v) := \phi(v) + \psi(e,v) - \phi(u)$.
Furthermore, for every $e=uv \in E(G)$ with $u \in X$ but $v \notin X$,
we relabel
$\psi(e,u) := \psi(e,u) - \phi(v)$ and $\psi(e,v) = \phi(v) + \psi(e,v)$. 
It is easy to check that, after the above relabeling, for every closed walk
$C$ it does not change whether $\psi(C) = 0$ or not, while $\psi(e,v) = 0$ for every $e \in E(G-X)$ and $v \in e$.

\paragraph{Extending a labeling of $X$.}
We now observe that, given a labeling $\phi_0 : X \to \Gamma$, finding a set $Z \subseteq V(G) \setminus X$ such that $\phi_0$ extends to a consistent labeling of $G-Z$ reduces to \textsc{Multiway Cut}.
\begin{lemma}\label{lem:gfvs:mwc}
There exists a randomized FPT algorithm with running time bound $2^{k^{\Oh(1)}} n^{\Oh(1)}$
that, given an untangled instance $(G,\psi,\weight,k,W,X)$ and a function $\phi_0 : X \to \Gamma$,
  checks if there is a set $Z \subseteq V(G) \setminus X$ of cardinality at most $k$ and weight
  at most $W$ such that $G-Z$ admits a consistent labeling extending $\phi_0$.
\end{lemma}
\begin{proof}
First, we check if for every $e=uv \in E(G[X])$ we indeed have $\phi_0(v) = \phi_0(u)+\psi(e,u)$,
as otherwise the answer is \no.
We construct a \textsc{Multiway Cut} instance as follows. 
Let $T$ be the set of those elements $g \in \Gamma$ such that there exists $uv \in E(G)$,
$u \in X$, $v \notin X$, and $g = \phi_0(u) + \psi(uv, u)$
(i.e., in a consistent labeling extending $\phi_0$, we would need to assign $g$ to $v$). 
Note that $|T| \leq |E(G)|$. 
Let $H$ be the graph consisting of a copy of $G-X$ (with weights inherited),
the set $T$ as additional vertices,
and for every $uv \in E(G)$, $u \in X$, $v \notin X$, an edge from $\phi_0(u)+\psi(uv,u) \in T$
to $v$. 
A direct check shows that it suffices to solve the obtained \textsc{Multiway Cut} instance
$(G,T,\weight,k,W)$ and return the answer (the proof of the equivalence is spelled out
    in the proof of Lemma~7 in~\cite{CyganPP16}). 
\end{proof}

\paragraph{Enumerating reasonable labelings of $X$.}
Since $\Gamma$ can be large, we cannot enumerate all labelings $\phi_0 : X \to \Gamma$.
In~\cite{CyganPP16}, a procedure is presented that enumerates
a family of $2^{\Oh(k \log k)}$ labelings such that, for every solution $Z$, there is a consistent
labeling of $G-Z$ that extends one of the enumerated labelings.

The main trick lies in the following reduction step. 
For $v \in V(G) \setminus X$ and $x \in X$, we define a flow graph $F(v,x)$ as follows.
Let $\Gamma_x$ be the set of those $g \in \Gamma$ such that there exists $xu \in E(G)$,
$u \notin X$ and $\psi(xu,u) = g$. Note that $|\Gamma_x| \leq |E(G)|$.
The graph $F(v,x)$ consists of a copy of $G-X$, the set $\Gamma_x$ as additional vertices
and, for every $xu \in E(G)$ with $u \notin X$, an edge $u \psi(xu,u)$. 

We have the following statement.
\begin{lemma}[Lemma~8 of~\cite{CyganPP16}]\label{lem:gfvs:red}
If there are $k+2$ paths in $F(v,x)$ from $v$ to distinct elements of $\Gamma_x$
that are vertex-disjoint except for $v$, then $v$ is contained in every solution
of cardinality at most $k$.
\end{lemma}
The condition of Lemma~\ref{lem:gfvs:red} can be checked in polynomial time.
If such a vertex $v$ is discovered, we can delete it, decrease $k$ by one,
decrease $W$ by $\weight(v)$, and repeat the analysis. 

Fix $x,y \in X$, $x \neq y$.
An \emph{external path} from $x$ to $y$ is a path with endpoints $x$ and $y$
and all internal vertices in $G-X$; note that an edge $xy$ is also an external path.
Let $\Gamma(x,y)$ be the set of all elements $g \in \Gamma$ such that there exists an external
path $P$ from $x$ to $y$ with $\psi(P) = g$. 
We have also the following statement.
\begin{lemma}[Lemma~9 of~\cite{CyganPP16}]\label{lem:gfvs:prop}
If there is no vertex $v$ as in Lemma~\ref{lem:gfvs:red},
   but for some $x,y \in X$, $x \neq y$ we have $|\Gamma(x,y)| \geq k^3(k+1)^2 + 2$,
   then there is no solution of cardinality at most $k$.
\end{lemma}
The condition of Lemma~\ref{lem:gfvs:prop} can be again checked in polynomial time and,
if we find that $\Gamma(x,y)$ is too large for some $x,y \in X$, $x \neq y$, we return \no.

Otherwise, we enumerate resonable labelings $\phi_0 : X \to \Gamma$ as follows.
First, we guess how $X$ is partitioned into connected components of $G-Z$ for a hypothetical
solution $Z$; in every connected component, we can set $\phi_0$ independently.
Let $Y \subseteq X$ be a set of vertices guessed to be in the same connected component of $G-Z$;
note that necessarily $Y$ needs to live in the same connected component of $G$,
so $\Gamma(x,y) \neq \emptyset$ for every distinct $x,y \in Y$.
Fix $y \in Y$ and set $\phi_0(y) = 0$. 
Note that in a consistent labeling of $G-Z$ that assigns the value of $y$ to $0$, 
for $x \in Y \setminus \{y\}$ the value assigned to $x$ needs to be in $\Gamma(y,x)$
as a path $P$ from $y$ to $x$ in $G-Z$ has $\psi(P) \in \Gamma(y,x)$. 
By Lemma~\ref{lem:gfvs:prop}, there are only $\Oh(k^5)$ options for $\phi_0(x)$.
Overall, this gives $2^{\Oh(k \log k)}$ options for $\phi_0$, as desired.

This finishes the proof of Theorem~\ref{thm:gfvs}.

\section{Directed Subset Feedback Edge/Vertex Set}\label{sec:dsfvs}

This section is devoted to the proof of Theorem~\ref{thm:dsfvs}. 
As discussed in Section~\ref{sec:prelims}, we can restrict ourselves
to the edge-deletion version, that is, to the \textsc{Directed Subset Feedback Edge Set} problem.
Furthermore, we can assume that red edges are undeletable (of weight $+\infty$):
for every $e = (u,v) \in R$, we subdivide $e$, replacing it with a path $u \to x_e \to v$;
the edge $(u,x_e)$ becomes red and of weight $+\infty$, and $(x_e,v)$ is not red and inherits
the weight of $e$.

Let $\inst = (G,R,\weight,k,W)$ be the input instance.
Using iterative compression, we can assume we are given access to a set $X \subseteq V(G)$ of size
at most $k+1$ such that $G-X$ has no cycle involving a red edge.

Let $Z \subseteq E(G) \setminus R$. Observe that $G-Z$ has no cycle containing a red edge
if and only if for every $(u,v) \in R$, there is no path from $v$ to $u$ in $G-Z$. 
The latter condition is equivalent to $u$ and $v$ being in different strong connected components
of $G-Z$. We will use the above reformulations of the desired separation property 
interchangably.

Let $Z$ be a sought solution. 
We start with some branching steps.
First, we guess how the vertices of $X$ are partitioned between strong
connected components of $G-Z$.
We identify vertices of $X$ that are guessed to be in the same connected components of $G-Z$;
note that in the branch where the guess is correct, this does not change whether two vertices
of $G-Z$ are in the same strong connected component or not.
Henceforth, by somewhat abusing the notation, we can assume that the vertices of $X$ lie 
in distinct strong connected components of $G-Z$. 
We guess the order of $X$ in a topological ordering of the strong connected components of $G-Z$;
that is, we guess an enumeration of $X$ as $x_1,x_2,\ldots,x_{|X|}$ such that in $G-Z$
there is no path from $x_j$ to $x_i$ for $1 \leq i < j \leq |X|$.
Since initially $|X| \leq k+1$, there are $2^{\Oh(k \log k)}$ branches up to this point
and we retain the property $|X| \leq k+1$.

We now construct a \gdpc{} instance $\inst'$. 
We first construct a graph $H$ as follows.
We start from $2|X|+1$ copies of the graph $G$, denoted $G^a$ for $1 \leq a \leq 2|X|+1$. 
For $u \in V(G)$, let $u^a$ be the copy of $u$ in the graph $G^a$.
For every $1 \leq a < b \leq 2|X|+1$ and every $u \in V(G)$ we add an arc $(u^b,u^a)$. 
For every red arc $(u,v) \in R$ and every $1 \leq a \leq |X|$, we add an arc $(u^{2a},v^{2a+1})$ 
Finally, we introduce two new vertices $s$ and $t$ and, for every $1 \leq a \leq X$ and
$1 \leq b \leq 2|X|+1$ an arc $(s,x_a^b)$ if $2a \geq b$ and an arc $(x_a^b,t)$ if $2a < b$.

For every $e = (u,v) \in E(G) \setminus R$, we make a bundle $B_e$ consisting of all $2|X|+1$
copies of the arc $e$. We set $\weight(B_e) = \weight(e)$. 
This finishes the description of a \gdpc{} instance $\inst' = (H,s,t,\emptyset,\bundles,\weight,k,W)$ with no clauses. See Figure~\ref{fig:dsfvs}.

\begin{figure}[tb]
\begin{center}
\includegraphics{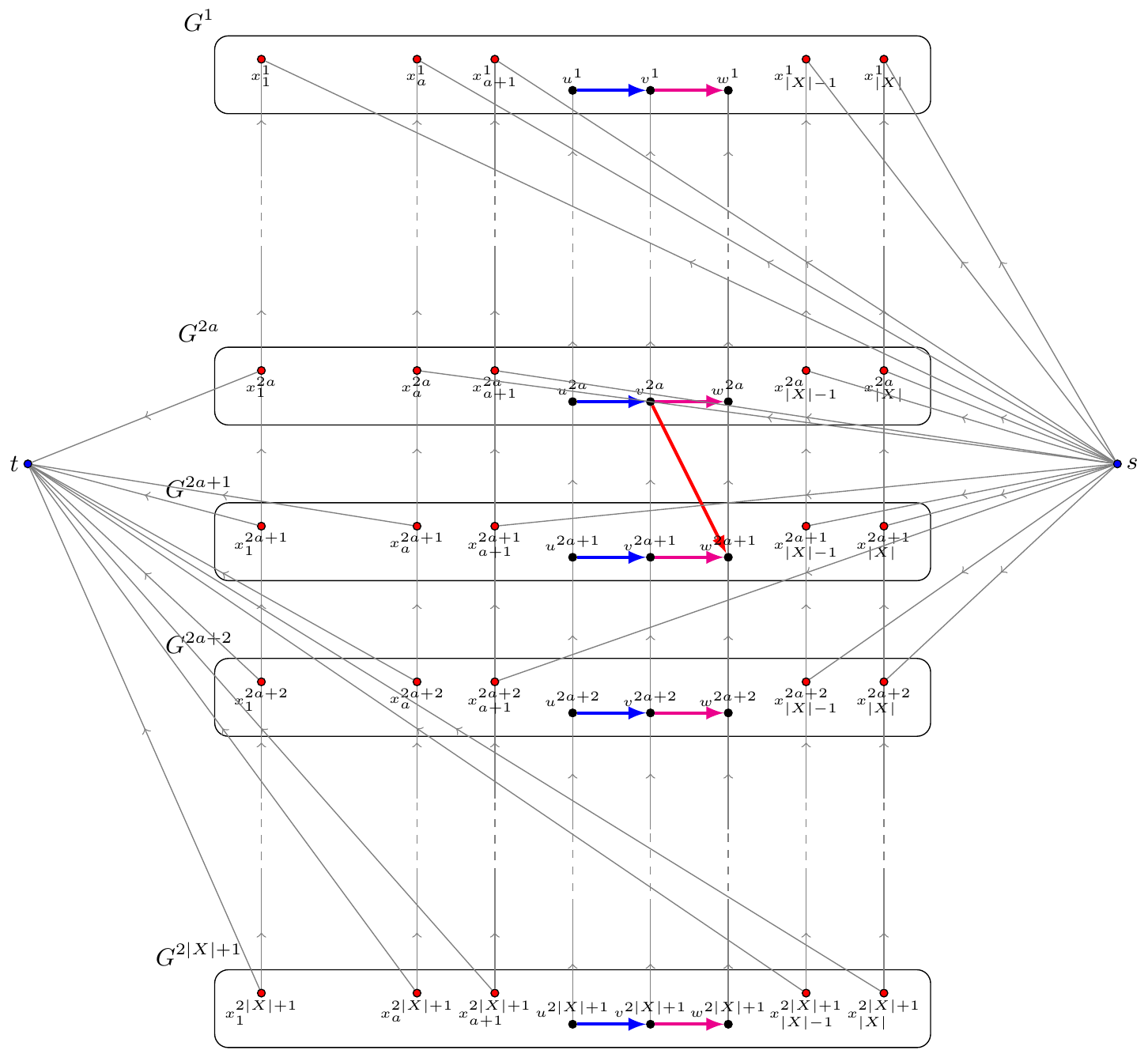}
\caption{Illustration of the reduction of Section~\ref{sec:dsfvs}.
All copies of an edge $e$ of $G$ form a bundle $B_e$: here blue edges
form one bundle for an edge $(u,v)$ and magenta edges
form another bundle for an edge $(v,w)$. 
Furthermore, if $(v,w)$ is red, then there is an extra
arc $(v^{2a},w^{2a+1})$ for every $1 \leq a \leq |X|$, depicted
in red. Intutively, this arc, together with
arcs $(s,x_a^{2a})$ and $(x_a^{2a+1},t)$ asks to destroy all closed walks that
pass both through $(v,w)$ and $x_a$.}\label{fig:dsfvs}
\end{center}
\end{figure}

We observe that the obtained instance has pairwise linked deletable edges, due to the existence
of crisp arcs $(u^b,u^a)$ for every $u \in V(G)$ and $1 \leq a < b \leq 2|X|+1$. 
Furthermore, every bundle contains at most $2|X|+1 \leq 2k+3$ deletable edges.
Thus, by Theorem~\ref{thm:gdpc:linked}, we can resolve it in randomized FPT time
$2^{\Oh(k^8 \log k)} n^{\Oh(1)}$. 

It remains to show that the answer to $\inst'$ is actually meaningful.
This is done in the next two lemmata that complete the proof of Theorem~\ref{thm:dsfvs}.

\begin{lemma}\label{lem:dsfvs:forward}
Let $Z \subseteq E(G) \setminus R$ be such that
$G-Z$ has no cycle containing a red edge, and, additionally,
all vertices of $X$ lie in distinct strong connected components of $G-Z$ 
and there is no path from $x_j$ to $x_i$ in $G-Z$ for every $1 \leq i < j \leq |X|$.
Then $Z' = \bigcup_{e \in Z} B_e$ is a solution to $\inst'$.
\end{lemma}
\begin{proof}
By contradiction, assume that $G-Z'$ contains a path $P'$ from $s$ to $t$.
Pick such a path $P'$ that minimizes the number of indices $a$ such that $P'$ contains a vertex
of $G^a$.
Let $a \in \{1, \ldots,2|X|+1\}$ be the minimum index such that $P'$ contains a vertex of $G^a$ and let $u^a \in V(P')$
be the last vertex of $P'$ in $G^a$. 
Observe that if $(s,x_i^b)$ is the first edge of $P'$, then $H$ also contains 
crisp edges $(s,x_i^{b'})$ for every $1 \leq b' \leq b$. Hence, by the minimality of $a$,
we can modify $P'$ so that the entire prefix from $s$ to $u^a$ is contained in $G^a$: 
whenever $P'$ traverses a vertex $v^{a'}$, we instead traverse the vertex $v^a$. 

Symmetrically, by choosing $a$ to be maximum such that $P'$ contains a vertex of $G^a$
and $u^a \in V(P')$ to be the first vertex of $P'$ in $G^a$, we observe that we can replace
the suffix of $P'$ from $u^a$ to $t$ so that it is completely contained in $G^a$. 

Observe that the only edges of $H$ that lead from $G^a$ to $G^b$ for $a < b$
are edges of the form $(u^{2a},v^{2a+1})$ for $1 \leq a \leq |X|$ and red arcs $(u,v)$.
Hence, we can assume that the path $P'$ is of one of the following two types:
\begin{enumerate}
\item All internal vertices of $P'$ lie in the same graph $G^a$.
\item For some $1 \leq a \leq |X|$, the path $P'$ first goes from $s$ via $G^{2a}$, then
uses one edge $(u^{2a},v^{2a+1})$ for some $(u,v) \in R$, and then continues via $G^{2a+1}$
to $t$.
\end{enumerate}
In the first case, let $(s,x_j^a)$ be the first edge of $P'$ and let $(x_i^a,t)$ be the
last edge of $P'$. By construction of $H$, we have $2j \geq a > 2i$, so $j > i$. 
Thus, $P'$ without the first and the last edge
gives a path in $G-Z$ from $x_j$ to $x_i$ for some $j > i$, a contradiction.

In the second case, let $(s,x_j^{2a})$ be the first edge of $P'$ and let $(x_i^{2a+1},t)$ be the
last edge of $P'$. By construction of $H$, we have $2j \geq 2a$ and $2a+1 > 2i$, so $j \geq i$. 
If $j > i$, $P'$ without the first and the last edge gives a path from $x_j$ to $x_i$, again a contradiction as in the first case. 
If $j = i$, then the subpath of $P'$ from $v^{2a+1}$ to $x_i^{2a+1}$ gives a path from
$v$ to $x_i$ in $G-Z$ and the subpath of $P'$ from $x_j^{2a}$ to $u^{2a}$ gives a path
from $x_j$ to $u$ in $G-Z$. As $j=i$, this gives a path from $v$ to $u$ in $G-Z$, a contradiction
as $(u,v)$ is a red edge.
\end{proof}

\begin{lemma}\label{lem:dsfvs:backward}
Let $Z'$ be a cut in $\inst'$ and let $Z = \{e \in E(G) \setminus R~|~B_e \cap Z' \neq \emptyset\}$.
Then $G-Z$ contains no cycle containing a red edge.
\end{lemma}
\begin{proof}
By contradiction, assume $G-Z$ contains a path $P$ from $v$ to $u$ for some $(u,v) \in R$.
Since $G-X$ contains no such path, $P$ contains a vertex of $X$. Let $x_i \in V(P)$.
Let $P_v$ be the prefix of $P$ from $v$ to $x_i$ and
let $P_u$ be the suffix of $P$ from $x_i$ to $u$.
Consider the copy $P_v^{2i+1}$ of $P_v$ in $G^{2i+1}$ and the copy $P_u^{2i}$ of $P_u$
in $G^{2i}$. Then, since $Z'$ contains no edge of $B_e$ for any $e \in E(P)$, 
$P_v^{2i+1}$ and $P_u^{2i}$ are disjoint with $Z'$. 
This is a contradiction, as a concatenation of $(s,x_i^{2i})$, $P_u^{2i}$, $(u^{2i}, v^{2i+1})$, 
$P_v^{2i+1}$, and $(x_i^{2i+1},t)$ is a path from $s$ to $t$ in $H-Z'$.
\end{proof}

\section{Conclusions}\label{sec:conc}

We showed fixed-parameter tractability of a number of weighted graph separation problems. Our first result extends a recent result of Galby et al.~\cite{GalbyMSST22WG}, 
who considered the special case of weighted \textsc{Multicut} in trees. 
For all our algorithms, we revisited an old combinatorial approach to the problem, adjusted it to weights,
and provided a reduction to \gdpc{} in one of its tractable variants. 
The application of the technique of flow-augmentation is hidden in the algorithms
for \gdpc{} (Theorems~\ref{thm:gdpc:2k2} and~\ref{thm:gdpc:linked}). 

We would like to highlight here one graph separation problem that resisted our attempts:
\textsc{Directed Symmetric Multicut}. Here, the input consists of a directed graph $G$,
weights $\weight : E(G) \to \mathbb{Z}_+$
(that is, we consider an edge-deletion variant, but, as we are working with directed graphs,
 it is straightforward to reduce between edge- and vertex-deletion variants),
integers $k$ and $W$, and a set $\terms \subseteq \binom{V(G)}{2}$
of unordered pairs of vertices of $G$.
The problem asks for an existence of a set $Z \subseteq E(G)$ of size at most $k$
and total weight at most $W$ such that for every $uv \in \terms$, the vertices $u$
and $v$ are not in the same strong connected component of $G-Z$
(i.e., $Z$ cuts all paths from $u$ to $v$ or cuts all paths from $v$ to $u$). 
Eiben et al.~\cite{EibenRW22IPEC} considered the parameterized complexity
of \textsc{Directed Symmetric Multicut} and gave partial results, 
but the main problem of the parameterized complexity of
\textsc{Directed Symmetric Multicut} parameterized by $k$ remains open,
even in the unweighted setting.

To motivate the \textsc{Directed Symmetric Multicut} problem further,
we point out that it has a very natural reformulation in the context of \emph{temporal CSPs},
that is, constraint satisfiaction problems with domain $\mathbb{Q}$ and 
access to the order on $\mathbb{Q}$.
More formally, a \emph{temporal CSP relation} is an FO formula with a number of free variables
that can be accessed via comparison predicates $x = y$, $x \neq y$, $x < y$, and $x \leq y$. 
A \emph{temporal CSP language} is a set of temporal CSP predicates.
For a temporal CSP language $\Lambda$, an instance of \textsc{CSP($\Lambda$)}
consists of a set of variables $X$ and a set $\mathcal{C}$ of constraints; each
constraint is an application of a formula from $\Lambda$ to a tuple of variables from $X$.
The goal is to find an assignment $\alpha : X \to \mathbb{Q}$ that satisfies 
all constraints.
In the \textsc{Max SAT($\Lambda$)} problem, we are additionally given an integer $k$
and the goal is satisfy all but $k$ constraints (i.e., delete at most $k$ constraints
    to get a satisfiable instance). 

In various CSP contexts, the \textsc{Max SAT($\Lambda$)} problem is usually hard,
yet the parameterized complexity landscape with $k$ as a parameter is often rich;
see e.g. the recent dichotomy for the Boolean domain~\cite{dfl-csp} and references therein.
The P vs NP dichotomy for temporal \textsc{CSP($\Lambda$)} is known
since over a decade~\cite{BodirskyK10}. Can we establish parameterized complexity
dichotomy for temporal \textsc{Max SAT($\Lambda$)} parameterized by $k$?

One of the most prominent examples of a temporal CSP languages is
$\Lambda = \{x=y, x \neq y, x < y, x \leq y\}$, called a \emph{point algebra}. 
Here, \textsc{CSP($\Lambda$)} is known to be polynomial-time solvable. 
We observe that \textsc{Max SAT($\Lambda$)} is equivalent to (unweighted)
\textsc{Directed Symmetric Multicut}.

In one direction, given an unweighted \textsc{Directed Symmetric Multicut} instance
$(G,\terms,k)$, we set $X = V(G)$, model every arc $(u,v) \in E(G)$ as 
a constraint $u \leq v$ and each pair $uv \in \terms$ as $k+1$ copies
of a constraint $u \neq v$. Intuitively, a desired
assignment $\alpha : V(G) \to \mathbb{Q}$ maps all vertices of the same strong
connected component to the same number, and otherwise sorts the strong connected components
according to a topological ordering. 

The other direction is slightly more involved due to some technicalities.
First, we replace each constraint $x=y$ with a pair of constraints $x \leq y$
and $y \leq x$; note that we will never want to delete both such constraints.
Similarly, we replace each constraint $x < y$ with $x \neq y$ and $x \leq y$;
again we will never want to delete both resulting constraints. Thus, we can assume that
the instance uses only $x \neq y$ and $x \leq y$ constraints. 
Then, for every constraint $x \neq y$, we introduce fresh copies $x'$ and $y'$ of $x$ and $y$,
introduce constraints $x \leq x'$, $x' \leq x$, $y \leq y'$, $y' \leq y$,
and $k+1$ copies of $x' \neq y'$, and delete $x \neq y$. Now deleting $x \neq y$
is equivalent to deleting one of the inequalities, say $x \leq x'$, and setting $x'$
to some very small number different than $y$ and $y'$. 
Thus, we end up in an instance where only $x \leq y$ and $x \neq y$ constraints are present,
and the latter constraints are always undeletable (appear in batches of $k+1$ copies).
Now, we can directly model it as \textsc{Directed Symmetric Multicut}: 
we set $V(G) = X$, for every constraint $x \leq y$ we add an arc $(x,y)$
and for every batch of $k+1$ constraints $x \neq y$ we add a pair $xy$ to $\terms$.

With a very similar reduction we observe that for $\Lambda' = \{x < y, x \leq y\}$
the problem \textsc{Max SAT($\Lambda'$)} is equivalent to (unweighted)
\textsc{Directed Subset Feedback Edge Set}:
every constraint $x < y$ is equivalent to a red arc $(x,y)$ 
and every constraint $x \leq y$ is equivalent to a non-red arc $(x,y)$. 

Therefore, the unresolved status of the parameterized complexity
of \textsc{Directed Symmetric Multicut} stands as the main obstacle
to obtain a dichotomy for parameterized complexity of \textsc{Max SAT($\Lambda$)}
for temporal CSP languages $\Lambda$, parameterized by the deletion budget $k$.

\bibliographystyle{plainurl}
\bibliography{refs}

\end{document}